\documentclass[10pt,journal,compsoc]{IEEEtran}
\usepackage[dvips]{graphicx}
\usepackage{amsmath,bm}
\usepackage{algorithmic}
\usepackage{algorithm}
\usepackage{url}
\usepackage{amssymb}
\usepackage{amsthm}
\usepackage{tabularx}
\usepackage{cases}
\usepackage{mathtools}
\usepackage{bbm}
\usepackage{xcolor}
\usepackage{color}
\usepackage{pbox}
\usepackage[utf8]{inputenc}
\usepackage{nicefrac}
\usepackage{amsmath}  

\usepackage{soul}
\newtheorem{remark}{Remark}
\newtheorem{prop}{Proposition}

\theoremstyle{definition}
\newtheorem{definition}{Definition}

\ifCLASSOPTIONcompsoc
\usepackage[caption=false,font=footnotesize,labelfont=sf,textfont=sf]{subfig}
\else
  \usepackage[caption=false,font=footnotesize]{subfig}
\fi

\begin{document}
\title{Joint Security-vs-QoS Game Theoretical Optimization for Intrusion Response Mechanisms for Future Network Systems}
%
%
%

 \author{Arash~Bozorgchenani,~\IEEEmembership{Member,~IEEE,}
        Charilaos~C. Zarakovitis,~\IEEEmembership{Member,~IEEE,}
        Su~Fong Chien,
        Qiang~Ni,~\IEEEmembership{Senior~Member,~IEEE,}   Antonios~Gouglidis,
        Wissam~Mallouli,
        Heng~Siong Lim,~\IEEEmembership{Senior~Member,~IEEE}
\IEEEcompsocitemizethanks{ 
 \IEEEcompsocthanksitem Arash~Bozorgchenani, Qiang~Ni and Antonios~Gouglidis are with the School of Computing and Communications, Lancaster University, UK, email:\{a.bozorgchenani, q.ni, a.gouglidis\}@lancaster.ac.uk.%

 \IEEEcompsocthanksitem Charilaos C. Zarakovitis is with National Center For Scientific Research "Demokritos", Greece, e-mail:c.zarakovitis@iit.demokritos.gr.%
 \IEEEcompsocthanksitem Su Fong Chien is with Advanced Intelligence Lab MIMOS Berhad Jalan Inovasi 3, TPM, 57000 Kuala Lumpur, Malaysia, e-mail:sf.chien@mimos.my. 
 \IEEEcompsocthanksitem Wissam Mallouli is with Montimage EURL, France, e-mail:  wissam.mallouli@montimage.com.%
 \IEEEcompsocthanksitem Heng Siong Lim is with Faculty of Engineering and Technology, Multimedia University, Jalan Ayer Keroh Lama, 75450 Melaka, Malaysia, e-mail:  hslim@mmu.edu.my. }%
 \thanks{This research is supported by the H2020 SANCUS project under agreement number 952672. More details can be found at: https://www.sancus-project.eu.}
 }

\IEEEtitleabstractindextext{
\begin{abstract}
Network connectivity exposes the network infrastructure and assets to vulnerabilities that attackers can exploit. Protecting network assets against attacks requires the application of security countermeasures. Nevertheless, employing countermeasures incurs costs, such as monetary costs, along with time and energy to prepare and deploy the countermeasures. Thus, an Intrusion Response System (IRS) shall consider security and QoS costs when dynamically selecting the countermeasures to address the detected attacks. This has motivated us to formulate a joint Security-vs-QoS optimization problem to select the best countermeasures in an IRS. The problem is then transformed into a matching game-theoretical model. Considering the monetary costs and attack coverage constraints, we first derive the theoretical upper bound for the problem and later propose stable matching-based solutions to address the trade-off. The performance of the proposed solution, considering different settings, is validated over a series of simulations.
\end{abstract}

 }

\maketitle

\section{Introduction}
\IEEEPARstart{B}{y} blending different types of technologies and advances, 5G offers various types of services such as smart home, vehicular communication, smart parking, air-ground integrated communication, fog/edge computing, industry 4.0, and blockchain-based services to name some \cite{SurveySec18}. Even though the new technologies pave the way for a fully connected people and things era by enabling many 5G services with various demands such as eMBB, mMTC, and uRLLC, they introduce new security challenges too \cite{BozorgchenaniARES22}. On one hand, this includes the utilization of 5G enabling technologies such as software-defined networking, network function virtualization, mobile edge computing, network slicing, etc. On the other hand, the heterogeneity of the 5G network brings new security challenges too, including the internet of things and end-user devices, service requests, new stakeholders and mission-critical applications, etc. \cite{SurveyTutorial20}. Needless to say that the pre-5G security threats still need to be addressed as well.



Cyberattacks target the network infrastructure to undermine the services' availability, and information confidentiality and integrity. The continuous rise in the number and complexity of attacks made it difficult to keep track of the large number of alerts generated by Intrusion Detection Systems (IDSs) and made security teams worldwide seek effective remediation tools. Potential tools to counteract ongoing threats are the Intrusion Reaction Systems (IRSs), which are capable of reacting against suspicious activities in real or near real-time by continuously monitoring the IDS alerts \cite{IRR_14}.
These reactions in a 5G network can include any of the atomic countermeasures of notifying the network operator, notifying the vendor, filtering the traffic, re-launching a node, re-configuring a virtual network function, replacing one node with another, providing a patch to prevent/remedy the identified attacks, and etc. The effectiveness of different countermeasures can be evaluated by their ability to mitigate the risk the assets of the network are exposed. One solution is applying a combination of different atomic countermeasures to address the affected nodes. On the one hand, from the orchestrator/provider's side, it is essential to address as many detected attacks as possible to minimize the impact of the threats in the network. On the other hand, applying countermeasures can have some potential consequences in terms of costs too. The remediation actions affect the system's QoS requirements such as time, energy, and monetary costs required to prepare and deploy the countermeasures, to name a few.

Thus, there exists a trade-off between maximizing the network security level and minimizing the QoS costs. Thereby, we formulate a joint security-vs-QoS countermeasure selection problem to optimize the Intrusion Response Mechanisms (IRMs) in network systems. To address this problem, we propose a game-theoretical approach to select the best set of countermeasures to be taken. Such a selection needs to balance the inherent trade-off between the effectiveness of the risk mitigation policy and its potential negative QoS impact. Such a balance is performed by the security administrator that has to maintain an adequate level of protection with a limited budget. Our {contribution}\footnote{This research has been conducted as part of the H2020 SANCUS project whose architecture and engines are elaborated in \cite{SANCUS_ARES21}.} can be summarized as follows:
\begin{enumerate}
    \item Formulating a novel joint security-vs-QoS problem for the optimal selection of countermeasures considering time, energy, and monetary cost as the QoS factors. To the best of our knowledge, the problem formulation has not been attempted by the relevant studies;
    \item Transforming the problem into a single objective problem by an $\epsilon$-constraint method and reformulating the problem in a two-sided matching game;
    \item Designing a Hospital/Resident (HR) model for the problem and proposing two Stable Matching (SM) countermeasure-oriented and attack-oriented algorithms to solve the problem from the security and QoS perspectives, which are unique.
    \item Deriving the upper bound for the problem by employing the decomposition and Branch-and-Bound (BB) techniques;
    \item Performing extensive simulation experiments to demonstrate the impact of different parameters on the performance of the two proposed  algorithms.
\end{enumerate}
The rest of the paper is organized as follows. In Section~\ref{sec:ReWorks}, we review the state of the art. In Section~\ref{sec:ProbFor} the model is described. Section~\ref{sec:solution} first shows the reformulated problem based on the HR model and later introduces our proposed algorithms. In Section~\ref{sec:numericalresult}, we present the simulation results. Section~\ref{sec:conclusion} concludes the paper.



\section{Related Works}
\label{sec:ReWorks}
In this section, we have conducted a thorough literature review and presented the most related works in the area of countermeasures selection for IRSs.

The existing schemes for selecting
the countermeasure to balance the attack damage and response cost can be roughly divided into two categories:
single countermeasure selection and multi-countermeasure selection.
In \cite{NADEEM2014, Nice13}
, the authors considered intrusion cost and the impact of the countermeasures to select a countermeasure against the intrusions. However, a single countermeasure is efficient in single-path intrusion and it cannot cut off attacks on multiple paths for multi-path intrusions. One efficient scheme for the problem is to select multiple countermeasures concurrently to reduce the potential risks and maximize the overall response utility.

There have been many studies focusing on the Cyber-Physical Systems (CPS) domain.
To name some, an Autonomous Response Controller (ARC) is presented in \cite{KHOLIDY21} to react against cyberattacks with a focus on Cyber–Physical Power Systems. 
The ARC can autonomously evaluate the security improvement resulting from applying certain remediations and it covers the uncertainty of the IDS alerts by using the Competitive Markov Decision Processes. 
However, their cost model only includes the cost of CPU and RAM.
Authors in \cite{DefenceTree20} developed a method to achieve minimum cost defense in the context of CPS. Specifically, such a procedure chooses optimal defense nodes using their developed Atom Attack Defense Trees (A2DT), which is a variant of the more conventional Attack-Defense Tree (ADT) model. Then, the authors used an ad-hoc methodology to solve the path calculation over the A2DT. 


Risk reduction requires the definition and implementation of a security configuration by the deployment of various security mitigation actions to reduce the risk. Hence, an optimization problem should be solved to select the most effective yet cost-efficient security countermeasures. Many researchers apply approximate bio-inspired solutions like Genetic Algorithms (GAs) to approach this problem. In \cite{Bio21}, the authors proposed an Artificial Immune System (AIS) to select countermeasures to defeat cyberattacks through cloning and mutation phases. They, however, suggested a context-aware stop condition based on experimental outcomes and authors’ subjective beliefs. 

A methodology to generate response policies is presented in \cite{ResponseICC20} addressing four problems of countermeasure selection, countermeasure deployment, the order of deployment, and the duration they last.
The authors proposed a decision-making framework for IRS that optimizes the responses based on some attributes and 
proposed a GA with Three-dimensional Encoding to solve the problem. 
However, solving the four problems altogether by a GA takes a long time to converge considering all the random options that an individual can take. 
All these studies applying evolutionary-based methods accept the risk of receiving only near-optimal solutions after going through many iterations.

There have also been a few Machine Learning (ML)-based solutions for the countermeasure selection problem in the literature. Authors in \cite{DRLStationaryIRS} studied the applicability of Deep Reinforcement Learning (DRL) for intrusion response control on stationary systems. This work was later extended to a non-stationary system in \cite{HybridIRS2020}, with a reward function based on execution time and cost of the executed actions. 
Experiments compare the proposed DRL algorithm with a Q-learning solution to demonstrate its feasibility.

There have been many studies on both the private and social costs of countermeasures. They focus on finding the upper-bound a risk-neutral firm should invest in cybersecurity \cite{Lelarge2012}, estimating the uncertain risk faced by an organization under cyberattack \cite{REES2011}, presenting a model to analyze optimal cybersecurity investment in supply chain and firms \cite{SIMON2020}, investigating the optimal balance between the prevention and detection and containment safeguards to deal with the uncertainty of cybersecurity \cite{PAUL2019}, to name some.


Several studies also considered graph-based modeling for the attacks and countermeasures. In \cite{Multi_path20}, a framework to respond to multi-path attacks is formulated and presented, which appears to be NP-hard. To resolve such a problem, they proposed a greedy algorithm to select the most appropriate countermeasures in a cost-sensitive way. The authors leveraged the Probabilistic Attack-Response Tree models to represent potential attacker movements and evaluate three metrics of security benefit, deployment cost, and negative impact. 
Similarly, in \cite{DAG20}, authors
relied on the ADT formalized with Directed Acyclic Graphs and then extracted from an ADT its defense semantics describing how the attacker and defender may interact. 
The authors developed an open-source tool to automate the described methodology.

A risk assessment methodology based on the application of an Attack Graph (AG) was proposed in \cite{Orly2021}, enhancing the standard AG-based model. 
Later, a heuristic approach is introduced to compute the optimal countermeasure for deployment while minimizing the overall risk with specific budget constraints. 
For a more comprehensive literature review on IRSs, you can read the work in \cite{SurveyCounter} which analyzed the major reaction proposals from 2012 to 2017, focusing on their principal advantages and potential deficiencies.

The described works take important steps within the reaction strategies ecosystem. However, there are several downsides present in the literature such as a) simplified modeling, b) approximate solutions, c) convergence issues in the learning-based solutions and the accuracy of the data used to train, and d) lack of consideration of different QoS parameters.  Moreover,
another important consideration in some of the studies in the literature is that the reaction frameworks are applied to specific scenarios leveraging a comprehensive knowledge of the protected system. One could argue that, in order to be generic and applicable to several contexts, 
the network dynamicity in the selection of countermeasures should be reflected in the model and results. In other words, the security-vs-QoS trade-off should be better reflected in the model and solution such that the operator/security administrator can make a wise decision at different time instants according to the network conditions, available resources, and the threat level. 

To this extent, the proposed solution in this work reflects the above shortcomings from the literature by formulating a joint security-vs-QoS optimization problem and strategically selecting the best remediation actions from an effective countermeasure repository by employing a stable matching game.

\section{System Model and Problem Formulation}
\label{sec:ProbFor}


In order to protect a network against attacks, it is vital to design IRSs and make appropriate response decisions to dynamically eliminate potential consequences, reduce security risks, and at the same time consider their impact on the QoS costs \cite{Nice13}. These remediations intend to protect the infrastructure and more specifically the network nodes/assets, which include, the IoT devices, base stations, servers, SDN controller, and network functions, denoted as $\mathcal{U}= \{ u_1, \dots, u_n, \dots u_N \}$. Each of these nodes can be attacked from different layers, i.e., hardware, firmware, operating system, application and etc. In the following, we introduce the security and QoS models separately and later formulate the joint security-vs-QoS problem for countermeasure selection. Please note that in the following sections, the terms mitigation action, remediation, and security countermeasures are used interchangeably.

 \subsection{Security Model}

Let us assume there exist $A$ types of attacks in the system, where $\omega_a$ shows attack type $a$ (e.g. DoS or eavesdropping).
For each attack type, we consider a mitigation action list that shows the possible countermeasure types that can be taken. To address an attack type across all the affected nodes, different instantiations of a countermeasure type might need to be deployed. For instance, \textit{the reconfiguration of a node for more robustness} is a countermeasure type where this reconfiguration can vary across different nodes. Hence, in the rest of the paper to facilitate ease of writing, the terms \textit{attack} and \textit{countermeasures} refer to attack types and countermeasure types, respectively.
Let us define $\mathcal{L}(a)$ as the list of countermeasures that can be taken for attack $a$ as 

\begin{equation}
    \label{CountermeasureList}
    \mathcal{L}(a)=\{ \theta_c | \mathbb{U}_c^{na}=1, \forall u_n \in \mathcal{U} \}
\end{equation}
where $\mathbb{U}_c^{na}$ is an indicator function which is 1 if countermeasure $c$ addresses the $a$-th attack on node $n$ and $\theta_c$ is the $c$-th countermeasure that can be taken for the $a$-th attack.
 Let us show $C$ as the total number of system countermeasures to address all attacks, i.e., $|\bigcup_{a=1}^A \mathcal{L}(a)|=C$.
 
Each countermeasure addresses at least one attack. Let us show the attacks the $c$-th countermeasure can address as
\begin{equation}
    \label{VulList}
    \mathcal{W}(\theta_c)=\{ \omega_a | \mathbb{U}_c^{na}=1, \forall u_n \in \mathcal{U} \}
\end{equation}
where, $|\mathcal{W}(\theta_c)|>0, \quad \forall \theta_c$.
On the other hand, one attack might affect different nodes across the network. We define the list of all of the attacks in the network (across all nodes) that a generic $c$-th countermeasure can address, as

\begin{equation}
    \label{Covered_UE_List}
    \mathcal{V}(\theta_c)=\{ v_n^a | \mathbb{U}_c^{na}=1, \forall u_n \in \mathcal{U} \}
\end{equation}
where $v_n^a$ represents an attack of type $a$ detected on the $n$-th node, that can be addressed by countermeasure $\theta_c$.

Let us define $\bar{\mathcal{L}}$ as the set of selected atomic countermeasures to address the detected attacks in the network, where $|\bar{\mathcal{L}}| \geq 0$. Then the total number of addressed attacks 
in the network is $
    \Big|\bigcup_{\theta_c \in \bar{\mathcal{L}}}  \mathcal{V}(\theta_c) \Big| $.
    
IDSs provide risk assessment metrics such as the severity and probability of the attacks. Exploiting this information the Risk Factor (RF) for the $a$-th attack can be driven as 
$R_{a}= S(a) \cdot P(a)$,
where $0 \leq P(a) \leq 1$ is the probability/likelihood of occurrence of an attack and $S(a) \in [0 \quad 10]$ is its severity. IDSs can also assess how much security is improved if a specific security enhancement is applied, which in turn assists the IRSs in relatively quantifying the effectiveness of different countermeasures \cite{KHOLIDY21}.
After taking a countermeasure both severity and probability matrices will be updated to see how effective the selected countermeasure is. In the rest of the paper, we only focus on the RF as it represents how severe and probable an attack is. Let us show $\bar{R}_{a}(\theta_c)$ as the updated RF of the $a$-th attack after taking the countermeasure and $R_{a}$ as the RF before taking the countermeasure. As part of the threat mitigation process, we would like to reduce the updated RF as much as possible by taking the most suitable countermeasure, hence we define $\Delta R_{a}(\theta_c)=R_{a}-\bar{R}_{a}(\theta_c)$ as the gap between the initial value of RF and the updated RF that should be maximized. It should be noted that $\Delta R_{a}(\theta_c) > 0$, i.e., the updated RF for those addressed attacks after taking a countermeasure is always reduced.

On the other hand, since the nodes in the network have different importance, we consider a priority-aware security utility function and define the overall security utility function for those selected atomic countermeasures as
\begin{equation}
    \label{Sec_Privacy}
    \frac{\sum_{\theta_c \in \mathcal{\bar{L}}} \sum_{ n=1}^N \sum_{ a=1}^A \alpha_n \Delta R_{a}(\theta_c) }{\sum_{ n=1}^N \sum_{a=1}^A \alpha_n R_{a}} 
\end{equation}
The nodes' coefficients ($0<\alpha_n \leq 1$) show the importance of each of the network nodes, e.g., an SDN controller has a higher coefficient than an IoT device. This ensures we prioritize reducing the RF for more important network nodes. Eq. \eqref{Sec_Privacy} calculates the weighted reduced RF of those addressed attacks across the nodes over the weighted initial RF values.

\subsection{Time and Energy Considerations}

The implementation of countermeasures exhausts some resources. For instance, the response \textit{dropping the malicious commands} consumes computer CPU and memory resources to analyze protocol data units of communication messages, along with consuming storage resources for recording all known attack signatures \cite{KHOLIDY21}. Hence, there will be some energy consumption and time spent in both the preparation and deployment phases of applying countermeasures. 
Assuming the countermeasures are deployed sequentially, we define time as the time duration spent for applying the $c$th countermeasure and formulate it as $
    T^{\text{tot}}(\theta_c)= T^{\text{pre}} (\theta_c)+ T^{\text{dep}}(\theta_c)$,
where $T^{\text{pre}}({\theta_c})$ is the time spent for the preparation of the countermeasure (officially termed as \textit{service preparation time}), and $T^{\text{dep}}({\theta_c})$ is the time spent for the deployment of the countermeasure (commonly termed as \textit{service deployment time}).  Countermeasure deployment can be manual or automatic (deployed by the system); however, here we focus on the automatic deployment of countermeasures. Thus, the overall \textit{time} for those selected atomic countermeasures can be written as

\begin{equation}
    \label{Tot_Time}
    \sum_{\theta_c \in \bar{\mathcal{L}}} T^{\text{tot}}(\theta_c)
\end{equation}

On the other hand, the total energy consumption by the system for the $l$th countermeasure can be written as $
    E^{\text{tot}}({\theta_c})= E^{\text{pre}}(\theta_c)+E^{\text{dep}}({\theta_c}) $,
where $E^{\text{pre}}({\theta_c})$ is the system energy consumption for preparation of the countermeasure and $E^{\text{dep}}(\theta_c)$ is the system energy consumption for the deployment of the countermeasure. Hence, we can define the overall energy consumption for those selected atomic countermeasures as

 \begin{equation}
    \label{Tot_Ene}
    \sum_{\theta_c \in \bar{\mathcal{L}}} E^{\text{tot}}(\theta_c)
\end{equation}

\subsection{Monetary Cost Consideration}

The defense cost is an important reference index in security countermeasure selection problems. For instance, the defense cost for the ADTree of a small network system with 15 attack nodes can reach \$300,000, which is a heavy burden for small and mid-sized enterprises \cite{DefenceTree20}. Thus, the monetary cost of a reaction (including fixed and variable costs) is an important metric, which can include hardware equipment, software development, labor, license, or loss resulting from users' dissatisfaction. 
In this regard, deprivation cost is also defined as the economic valuation of the post-disaster (i.e., cyber attacker) human suffering (i.e., attacked firms' loss) associated with a lack of access to a good/service \cite{DisasterModel13}. 
For instance, a DoS attack can cause a
degradation of service on an ISP’s network, resulting in service level agreements being violated. A cost could be reimbursements to customer networks. The same incident might lead to a loss of reputation for the ISP, which is a qualitative impact
\footnote{More detailed modeling can be considered to extend the monetary cost representation, however, this is out of the scope of this research.}.
Let us denote $\Psi^{\text{tot}}({\theta_c})$ as the monetary cost of taking $c$th countermeasure including the above-mentioned factors.

\subsection{Problem Formulation}

Having defined the security and QoS models, the joint security-vs-QoS utility function for an atomic countermeasure is defined as 
\begin{align}
    \label{Cost_fun}
    \Upsilon(\theta_c)= \frac{ \frac{ \sum_{ n=1}^N \sum_{ a=1}^A \alpha_n \Delta R_{a}(\theta_c) }{\sum_{ n=1}^N \sum_{a=1}^A \alpha_n R_{a}} }{\beta_1 T^{\text{tot}}(\theta_c)+\beta_2 E^{\text{tot}}(\theta_c) +\beta_3 \Psi^{\text{tot}}(\theta_c)}
\end{align}
where $\beta_*$ refers to the coefficient of the QoS parameters such that $\sum_{i=1}^3 \beta_i=1$. 
As the countermeasure selection problem is restricted by QoS costs, the IRS might not always be able to address all the attacks at once in a large network. On the other hand, best efforts should be made to minimize the assets' exposure to threats. The goal of the joint security-vs-QoS optimization problem is to optimize the IRMs by selecting the most suitable countermeasures in order to maximize \textbf{(a) the joint utility function, and (b) the number of addressed attacks across the nodes.} Hence, we define

\begin{align}
\label{Obj-fun2}
& \mathbf{P1}:  \underset{\bar{\mathcal{L}}}{\max} \Bigg\{
 \sum_{\forall \theta_c \in \bar{\mathcal{L}}} \Upsilon(\theta_c), \bigg| \bigcup_{\forall \theta_c \in \bar{\mathcal{L}}}   \mathcal{V}(\theta_c) \bigg| \Bigg\}
\end{align}%
subject to
\begin{align}
& \mathbf{C1.1}: \sum_{\forall \theta_c \in \bar{\mathcal{L}}} \Psi^{\text{tot}}({\theta_c})< \xi \label{c1.1}
\end{align}
%
The objective function \eqref{Obj-fun2} targets to find the best countermeasures to be selected in the decision vector $\bar{\mathcal{L}}$ to jointly maximize the utility function and the number of addressed attacks across the nodes. Constraint \eqref{c1.1} represents the maximum monetary budget for taking countermeasures. 

Problem $\mathbf{P1}$ is a bi-objective optimization problem. In order to solve the problem we employ an \textit{$\epsilon$-constraint} method. The \textit{$\epsilon$-constraint} method generates single objective sub-problems by transforming all but one objective into constraints \cite{BerubeEpsilon}. As our problem is a bi-objective optimization problem, this is a good method, as it can generate the exact Pareto front by varying the upper-bound of the new {constraint}\footnote{The impact of varying upper-bounds will be studied in the simulation results section}.
This method has been broadly used in the literature \cite{Alexandros20}. 
Hence, by following the \textit{$\epsilon$-constraint} approach we transform $\mathbf{P1}$ to $\mathbf{P2}$ as below:

\begin{align}
\label{Obj-fun3}
& \mathbf{P2}:  \underset{\bar{\mathcal{L}}}{\max} \Bigg\{
 \sum_{\forall \theta_c \in \bar{\mathcal{L}}} \Upsilon(\theta_c) \Bigg\}
\end{align}%
subject to
\begin{align}
& \mathbf{C2.1}: \sum_{\forall \theta_c \in \bar{\mathcal{L}}} \Psi^{\text{tot}}({\theta_c})< \xi \label{c2.1}\\
& \mathbf{C2.2}: \bigg| \bigcup_{\forall \theta_c \in \bar{\mathcal{L}}}   \mathcal{V}(\theta_c) \bigg| \geq \bar{M} \label{c2.2}
\end{align}

In $\mathbf{P2}$, a new bounded constraint $\mathbf{C2.2}$ is defined which was one of the objectives in $\mathbf{P1}$, indicating the number of addressed attacks across all nodes shall be larger than the threshold $\bar{M}$.

The optimization problem aims at taking the most suitable set of countermeasures from the mitigation action list in order to maximize the joint utility function of the system and addresses a minimum of a certain number of attacks by a maximum defined monetary budget.
Different coefficients for QoS parameters in \eqref{Cost_fun} enforce to outweigh some of the objectives (based on the network condition), which can be set dynamically at different time instants according to our priorities/preferences of the objectives.

\section{Many-To-One-Stable Matching Solution}
\label{sec:solution}
In this section, we propose assigning/matching the countermeasures to the attacks by a framework that considers stability as the solution concept instead of optimality.
The applied framework involves a two-sided matching game.
A Stable Matching Problem (SMP) is produced by a distributed process that matches together preference relations of the two sides that are of the same size. The order of preferences is given by the strictly ranked rate utilities of the two sides \cite{StableMatching15}. SM solutions have been broadly used in wireless networks for problem-solving \cite{SM_Wireless}.
In our problem, however, the number of detected attacks might be different from the number of countermeasures (i.e., different set sizes), which means we need to seek a many-to-one generalization of SMP called the HR problem \cite{STP_Thesis}. 

\subsection{Hospital/Residents Model}

In the HR problem, each hospital has one or more posts to be filled, and a preference list ranking a subset of the residents. Likewise, each resident has a preference list ranking a subset of the hospitals. The capacity of a hospital is its number of available posts. We need to match each resident to at most one hospital such that no hospital exceeds its capacity threshold while observing the stability conditions \cite{STP_Thesis}.
We can map the residents to the attacks and the hospitals to the countermeasures, and design an SM between the two sides in order to mitigate the attacks' impact on the system.

The SMP is modeled by the tuple $\big \langle \mathcal{A}, \mathcal{C}, \{ U_{a}\}_{a \in \mathcal{A}}, \{ U_{c}\}_{c \in \mathcal{C}},  \{ q_c\}_{c \in \mathcal{C}} \big \rangle$, where $\mathcal{A}$ is the set of attacks, $\mathcal{C}$ is the set of countermeasures, $\{ U_{a}\}$ and $\{ U_{c}\}$ are the utility functions of attacks and countermeasures, and $\{ q_c\}$ are the quotas associated with each countermeasure representing the maximum number of attacks they can address, where in our work is equal to the $\mathcal{W}(\theta_c)$, i.e., no limitations on capacity. Let us introduce the following definitions \cite{Matching_1990}:

\begin{definition}
\label{Def_Matching}
A Matching $M$ is from the set $\mathcal{A} \cup \mathcal{C}$ into the set of unordered family  of elements $\mathcal{A} \cup \mathcal{C}$ such that:
\begin{enumerate}
    \item $|M(a)|=1, \forall a \in \mathcal{A}$
    \item  $1 \leq |M(c)| \leq q_c, \forall c \in \mathcal{C}$
    \item $M(a)=c$ if and only if $a \in M(c)$.
\end{enumerate}
\end{definition}
In Definition \ref{Def_Matching}, the first criterion means each attack (resident) is matched to one countermeasure (hospital), and the second one means each countermeasure has a maximum capacity of $q_c$ as the number of attacks it can address, and the last criterion means a countermeasure $c$ is the match for attack $a$, iff the attack $a$ is in the preference list of countermeasure $c$ (i.e., $a$ is acceptable to $c$). It should be noted that we have set $q_c=\mathcal{W}(\theta_c)$ and this guarantees that no countermeasure will be over-subscribed.

\begin{definition}
\label{Blocked_Matching}
The matching $M$ is blocked by the pair $(a,c) \in \mathcal{A} \times \mathcal{C}$ if the following conditions are satisfied
\begin{enumerate}
    \item $a$ and $c$ find each other acceptable
    \item $U_a(c) > U_a(M(a))$
    \item either $|M(c)| < q_c$ and $U_c(a) >0$
    \item or $U_c(a) >U_c(a^{\prime})$ for some $a^{\prime} \in M(c)$
\end{enumerate}
\end{definition}
According to Definition \ref{Blocked_Matching} if conditions (1), (2), and either of (3) or (4) occur that means either of the sides prefers each other over their current matching.

\begin{definition}
\label{Def_Stable}
A Matching $M$ is stable if it admits no blocking pair.
\end{definition}
The stability as a criterion for matching ensures that neither side of the game has the incentive to improve outside of the matching scheme.

Let us define the cost function of an attack as:
\begin{align}
    \label{U_V} 
    U_a(c)=& \beta_1 \Bigg( \frac{ T^{\text{tot}}_{ca}
     -x_{\text{min}}}{x_{\text{max}}-x_{\text{min}}} \Bigg)+\beta_2 \Bigg( \frac{ E^{\text{tot}}_{ca}
     -x_{\text{min}}}{x_{\text{max}}-x_{\text{min}}} \Bigg)+ \nonumber\\ & \beta_3 \Bigg( \frac{ \Psi^{\text{tot}}_{ca}
     -x_{\text{min}}}{x_{\text{max}}-x_{\text{min}}} \Bigg)
\end{align}

and the utility function of a countermeasure as
\begin{equation}
    \label{U_C}
    U_c(a)=\frac{\sum_{n=1}^{N_a} \alpha_n \Delta R_{ca}}{\sum_{n=1}^{N_a} \alpha_n R_a}
\end{equation}
where $x_{\text{max}}$ and $x_{\text{min}}$ denote the maximum and minimum value of the respective QoS parameter (provided in Table \ref{tab1}), $\sum_{i=1}^3 \beta=1$, and $N_a$ is the number of nodes with attack $a$. 

Let us define $x_{ca} \in \{0,1\}$ as a decision variable meaning if countermeasure $c$ and attack $a$ are matched. Then the problem $\mathbf{P2}$ can be reformulated in the form of an SMP as

\begin{align}
\label{Obj-fun4}
& \mathbf{P3}:  \underset{\mathbf{x}}{\max} \Bigg\{
 \sum_{c=1}^{C} \sum_{a=1}^{A}  \Bigg( \frac{U_c(a)}{U_a(c)} \Bigg) x_{ca} \Bigg\}
\end{align}%
subject to
\begin{align}
& \mathbf{C3.1}: \sum_{c=1}^{C} \sum_{a=1}^{A}  \Psi^{\text{tot}}({\theta_c}) x_{ca}< \xi \label{c3.1}\\
& \mathbf{C3.2}: \sum_{c=1}^{C} \sum_{a=1}^{A}  N_a x_{ca} \geq \bar{M} \label{c3.2}\\
& \mathbf{C3.3}: \sum_{c=1}^C x_{ca}=1, \quad \forall a \in \mathcal{A} \label{c3.3}\\
& \mathbf{C3.4}: x_{ca} \in \{0,1\} \label{c3.4}
\end{align}
where $\mathbf{x}$ is the matching decision vector identifying the selected atomic countermeasures. Constraint \eqref{c3.1} and \eqref{c3.2} represent the monetary cost and the minimum number of attacks to be addressed across all nodes. Constraint (\ref{c3.3}) assures that an attack is matched with only one countermeasure. Constraint \eqref{c3.4} indicates that a countermeasure and an attack are either matched or not (binary value). It should be noted that $U_a(c)>0 \quad \forall c$ in order to get a feasible solution. Please note the difference between $A$ and $\Bar{M}$, where the first shows the number of attacks and the latter the minimum number of addressed attacks across all nodes in the network.

\begin{remark}
    \label{R1}
    Weighting the two sides of the SMP (i.e, security and QoS) does not have any impact on the preference list formation. Hence, it does not yield different solutions in $\mathbf{P3}$.
\end{remark}

\begin{remark}
    \label{R2}
    The existence of different weights on each side of the game (if applicable) can result in different matching; hence, different solutions in  $\mathbf{P3}$.
\end{remark}

In order to solve the above SMP, we first present its upper bound through theoretical analysis, and later propose distributed solutions.

\subsection{Theoretical Analysis}
One of the most famous challenges in Combinatorial optimization is the Knapsack problem, which has been proven to be \textit{NP-Hard} \cite{IntoAlg}. One of the variants of the Knapsack problem is called the Multiple Knapsack Problem (MKP). In MKP, there exist multiple Knapsacks each with a certain capacity. The decision is whether an item should be selected and if yes, to which Knapsack it should be allocated to. 
$\mathbf{P3}$ resembles a Multiple Multi-dimensional Knapsack Problem (MMKP), where the two dimensions are \eqref{c3.1} and \eqref{c3.2} and $C$ represents the number of Knapsacks. As MMKP is also \textit{NP-hard}, similar to \cite{MMNP}, we derive the upper bound and discuss the exact solution by employing the decomposition and BB techniques as the dynamic alternative approaches require huge memory requirements.

As the main challenge in BB algorithm is the determination of the upper bound, we only focus on the derivation of the upper bound of $\mathbf{P3}$. The upper bound for the standard Knapsack problem has been calculated by greedy algorithms \cite{IntoAlg}. Hence we decompose the MMKP into several simple standard Knapsack problems and the upper bound of the original $\mathbf{P3}$, which is an MMKP, can be obtained by solving the sub-problems in parallel. We first relax two of the constraints in $\mathbf{P3}$ and rewrite it as

\begin{align}
\label{Obj-fun5}
& \mathbf{P4}: \underset{\mathbf{x}}{\max} 
  \mathcal{L}(x,\rho, \bm{\vartheta})= 
 \sum_{c=1}^{C} \sum_{a=1}^{A}  \Bigg( \frac{U_c(a)}{U_a(c)} \Bigg) x_{ca}\\ \nonumber & + \rho \Bigg( \sum_{c=1}^{C} \sum_{a=1}^{A} N_a x_{ca} - \bar{M} \Bigg) + \sum_{a=1}^A \vartheta_a \Big( \sum_{c=1}^{C} x_{ca}-1 \Big)
\end{align}%
subject to $\mathbf{C3.1}$ and $\mathbf{C3.4}$, 
where $\rho$ and $\bm{\vartheta}=[\vartheta_1,\dots\,\vartheta_A]$ are the dual variables associated with constraints $\mathbf{C3.2}$ and $\mathbf{C3.3}$, respectively. The optimum value of $\mathbf{P4}$ is an upper bound of the optimum value of $\mathbf{P3}$ for arbitrary non-negative $\rho$ and $\bm{\vartheta} \in R^A$. To further gain a tight upper bound, we have to optimize $\mathbf{P4}$ for the dual variables as

\begin{align}
\label{Obj-fun6}
& \mathbf{P5}: g(\rho, \bm{\vartheta})= \underset{\rho >0, \mathbf{ \bm{\vartheta} }}{\min} 
 \mathcal{L}(x,\rho, \bm{\vartheta} )
\end{align}%
subject to $\mathbf{C3.1}$ and $\mathbf{C3.4}$.
Considering $\mathbf{P4}$ we can rewrite $\mathcal{L}(x,\rho, \bm{\vartheta} )$ as

\begin{align}
\label{MMKP_Result}
 \mathcal{L}(x,\rho, \bm{\vartheta} )& = \sum_{c=1}^C \sum_{a=1}^A  \frac{U_c(a)}{U_a(c)} x_{ca} + \rho  \sum_{c=1}^C \sum_{v=1}^A  N_a x_{ca}\\ \nonumber & -\rho\Bar{M}+ \sum_{c=1}^C \sum_{a=1}^A  \vartheta_a x_{ca}-\sum_{a=1}^A \vartheta_a\\ \nonumber & = 
 \sum_{c=1}^C \Bigg\{ \sum_{a=1}^A \Bigg( \frac{U_c(a)}{U_a(c)}+\rho N_a+ \vartheta_a \Bigg) x_{ca}  \Bigg\}\\ \nonumber &- \rho \Bar{M}-\sum_{a=1}^A \vartheta_a
\end{align}%

Obviously, the upper bound of the original MMKP can be computed by decomposing equation \eqref{MMKP_Result} into $C$ standard Knapsack problems in parallel that can significantly reduce the computing power. In the simplest case, a sub-problem for each countermeasure $c$ can be written as the following minimization problem

\begin{align}
\label{Obj-fun7}
& \mathbf{P6}:\underset{\mathbf{\mathbf{x}}}{\min} \sum_{a=1}^A \Bigg( \frac{U_c(a)}{U_a(c)}+\rho N_a+ \vartheta_a \Bigg) x_{ca}
\end{align}%
subject to $\mathbf{C3.1}$ and $\mathbf{C3.4}$

Denote the minimum value of each $c$-th sub-problem as $\mu_c$, the summation of each minimum value of the $c$-th sub-problem of $\mathbf{P6}$ plus the last two terms in \eqref{MMKP_Result} gives the upper bound of the original MMKP in $\mathbf{P3}$ as

\begin{equation}
    \label{Sol}
    \sum_{c=1}^C \mu_c- \rho \Bar{M}- \sum_{a=1}^A \vartheta_a
\end{equation}
The solution to \eqref{Sol} can be calculated efficiently due to the greedy choice property possessed by the standard Knapsack problem; hence, convergence is guaranteed. Now we can rewrite $\mathbf{P5}$ as

\begin{align}
\label{Obj-fun8}
& \mathbf{P7}: g(\rho, \bm{\vartheta})= \underset{\rho>0, \mathbf{ \bm{\vartheta} }}{\min} 
      \Bigg(\sum_{c=1}^C \mu_c- \rho \Bar{M}- \sum_{a=1}^A \vartheta_a \Bigg)
\end{align}%
subject to $\mathbf{C3.1}$ and $\mathbf{C3.4}$.
Please note that the process of obtaining the optimal dual variables for $\mathbf{P7}$ and the rest of the BB algorithm follow the standard procedure, and thus will not be discussed here.

\subsection{Distributed Stable Matching-based Solution}

In this section, we propose two algorithms to solve the formulated SMP by considering the constraints in our problem. Each of these algorithms considers the preference of one side of the game. Hence, we introduce an Attack-oriented SM (ASM) algorithm and a Countermeasure-oriented SM (CSM) algorithm.

In order to respect the constraint $\mathbf{C3.2}$ in our SM solutions, we first consider a pre-processing step. As illustrated in Alg. \ref{Feas_Sol} among all the possible countermeasures sets (that form a solution) combination that can be taken for addressing the attacks, we select those solutions that can cover a minimum of a certain number of attacks in the network 
as feasible solutions. This allows us to ensure the constraint $\mathbf{C3.2}$ is always respected and the complexity of the SM solution algorithms will be reduced by solving the problem only for the feasible solutions instead of all solutions. Then these feasible solutions, $\mathcal{S}$, are passed to the SM algorithms to find the matching solutions. 

\renewcommand{\algorithmicrequire}{\textbf{Input:}}
\renewcommand{\algorithmicensure}{\textbf{Output:}}
\begin{algorithm}
\caption{Feasible Solution Formation}\label{Feas_Sol}
\footnotesize
\begin{algorithmic}[1]
\REQUIRE{$C$, $\bar{M}$}
\ENSURE{$\mathcal{S}$}
\FOR{each $i=1$ to $C$}
\STATE $Y \leftarrow$ all $\binom{C}{i}$ countermeasure solution combinations
\FOR{each solution $j$ in $Y$}
\IF{$\sum_{c \in Y_j} \sum_{a=1}^{A} x_{ca} N_a \geq \bar{M}$} 
\STATE $\mathcal{S} \leftarrow Y_j$
\ENDIF
\ENDFOR
\ENDFOR
\end{algorithmic}
\end{algorithm}
We adopt the \emph{Gale-Shapley} algorithms \cite{STP_Thesis,GaleShapely} and propose two algorithms which are namely ASM and CSM. The difference between the two algorithms lies in the fact that which side's preference is considered for the matching. The preference lists of ASM and CSM are composed using $\eqref{U_V}$ and $\eqref{U_C}$, respectively.
Alg. \ref{Vul_Alg} shows the ASM solution which finds the matching for each feasible solution in $\mathcal{S}$. The algorithm continues matching each attack with its highest preferences until the required number of attacks across the nodes is covered.


\renewcommand{\algorithmicrequire}{\textbf{Input:}}
\renewcommand{\algorithmicensure}{\textbf{Output:}}
\begin{algorithm}
\caption{Attack-oriented SM Algorithm}\label{Vul_Alg}
\footnotesize
\begin{algorithmic}[1]
\REQUIRE{Preference list of $\mathcal{V}$ and $\mathcal{C}$}
\ENSURE{ASM solution covering a minimum of $\bar{M}$ attacks}

\textbf{Initialization phase:}
\STATE initialize all of the $a \in \mathcal{A}$ and $c \in \mathcal{C}$ to be free


\textbf{Matching evaluation:}
\WHILE{$\Big(\sum_{c=1}^{C} \sum_{a=1}^{A} x_{ca} N_a\Big) <\bar{M}$}
\STATE $c \coloneqq$ first countermeasure on $a$'s list
\STATE $M=M \cup \{(a,c)\}$
\ENDWHILE

\end{algorithmic}
\end{algorithm}
In Alg.\ref{Count_Alg}, on the other hand, after a countermeasure proposes an attack and they are matched, any successor countermeasure is removed from the attack's list. This is due to the fact that the newly matched attack will not prefer any of the successor countermeasures in the future over the one to which it is matched (as it has lower preferences for them). This shortens the SM solution space, as each countermeasure does not propose those attacks to which it cannot match.
The CSM Algorithm is performed vertically, i.e., matching the first preference of each of the countermeasures, then their second preference and etc., as it results in better performance and it is more fair w.r.t. the horizontal matching when the coverage of the attacks is supposed to be performed partially, i.e., $\mathbf{C3.2}$. However, when the percentage of the covered attacks across the nodes is 100\% both perspectives (i.e., vertical and horizontal matching) result in the same solution.

\renewcommand{\algorithmicrequire}{\textbf{Input:}}
\renewcommand{\algorithmicensure}{\textbf{Output:}}
\begin{algorithm}
\caption{Countermeasure-oriented SM Algorithm}\label{Count_Alg}
\footnotesize
\begin{algorithmic}[1]
\REQUIRE{Preference list of $\mathcal{A}$ and $\mathcal{C}$}
\ENSURE{CSM solution covering a minimum of $\bar{M}$ attacks}

\textbf{Initialization phase:}
\STATE initialize all of the $a \in \mathcal{A}$ and $c \in \mathcal{C}$ to be free


\textbf{Matching evaluation:}
\WHILE{$\Big(\sum_{c=1}^{C} \sum_{a=1}^{A} x_{ca} N_a\Big) <\bar{M}$}
\STATE $a \coloneqq$ first attack on $c$'s list
\IF{$a$ is already assigned to $c^{\prime}$ }
\STATE $M=M \setminus \{(a,c^{\prime})\};$
\ENDIF
\STATE $M=M \cup \{(a,c)\}$
\FOR{each successor of $c^{\prime}$ of $c$ on $a$'s list}
\STATE delete the pair $(a,c^{\prime})$;
\ENDFOR
\ENDWHILE

\end{algorithmic}
\end{algorithm}


The overall steps for our proposed solution is shown in Alg.~\ref{Solution}. First the feasible solution set is formed, then the result of the SM (either ASM or CSM) is stored and after performing this process for all the feasible solutions, the one that respects $\mathbf{C3.1}$ and maximizes $\mathbf{P3}$ is the final solution.

\renewcommand{\algorithmicrequire}{\textbf{Input:}}
\renewcommand{\algorithmicensure}{\textbf{Output:}}
\begin{algorithm}
\caption{The proposed Solution}\label{Solution}
\footnotesize
\begin{algorithmic}[1]
\REQUIRE{$C$, $\bar{M}$, preference lists}
\ENSURE{Solution of $\mathbf{P3}$}
\STATE Run Alg.\ref{Feas_Sol} to obtain the set $\mathcal{S}$
\FOR{each $i$ in $\mathcal{S}$}
\STATE Run Alg.\ref{Vul_Alg} or to Alg.\ref{Count_Alg} to solve the SM problem
\STATE $\mathcal{M} \leftarrow M_i$
\ENDFOR
\STATE return  $ \operatorname*{argmax}_{M \in \mathcal{M}}\Bigg\{
 \sum_{c=1}^{C} \sum_{a=1}^{A} x_{ca} \Bigg( \frac{U_c(a)}{U_a(c)} \Bigg) \Bigg\}$ s.t. $\mathbf{C3.1}$

\end{algorithmic}
\end{algorithm}

The time complexity of Alg. \ref{Solution} is $O(C|Y|+|\mathcal{S}|CA)$. The complexity of Line 1 (i.e., Alg.\ref{Feas_Sol}) is $O(C|Y|)$. The complexity of Alg.\ref{Vul_Alg} or Alg.\ref{Count_Alg} is $O(CA)$, which shows the dimension of the preference lists \cite{SMComplexity}.
The SM algorithms provide a good sub-optimal solution as has been demonstrated in the literature \cite{LeshamJSAC}.

The outcomes of the two algorithms are not necessarily
equal. In the ASM solution, the attacks are allocated to their most preferred countermeasure as the countermeasures do not become over-subscribed, hence, the ASM algorithm gives the best QoS-wise SM solution. However, in the CSM solution, the countermeasures propose the attacks according to their preferences (i.e., security) and they are matched if this is the best proposal it has received (in terms of QoS). Hence, the solution of CSM is more balanced in terms of security and QoS.

\begin{prop}
    \label{P4Matching}
    ASM and CSM respect all the criteria of a matching game.
\end{prop}
\begin{proof}
    \label{P4Matching_Proof}
    We need to ensure the three criteria in Definition \ref{Def_Matching} are respected. $\mathbf{C3.3}$ guarantees that the first criterion in Definition \ref{Def_Matching} is respected for both ASM and CSM algorithms. Moreover, as we assume the number of attacks matched with a countermeasure cannot exceed the capacity of the countermeasure, criterion 2 is also respected. Finally, the third criterion in Definition \ref{Def_Matching} is also respected as each matching pair is performed following the preference lists of the two sides, i.e., they are in each other's preference list \textit{iff} the two sides are acceptable to each other. Thus, the proposition is proved.
\end{proof}

\begin{prop}
    \label{SM_Stable}
    The ASM and CSM algorithms are stable for full attack coverage.
\end{prop}
\begin{proof}
In order to prove the stability of the two algorithms, the blocking situations defined in Definition \ref{Blocked_Matching} need to be avoided. 
Suppose $(a_1,c_1)$ and $(a_2,c_2)$ are the results of the CSM algorithm. Let us assume $c_2$ prefers $a_1$ over its matching (which is $a_2$). This means $c_2$ must have proposed to $a_1$ before proposing to $a_2$ due to the functionality of the \textit{Gale-Shapley} algorithm. Since $c_2$ proposed to $a_2$ at some point, $a_1$ must have rejected $c_2$. This signifies at the time of rejection, $a_1$ preferred some $c^{\prime}$ over $c_2$. From the output of this matching example, it can be observed that $a_1$ has chosen $c_1$ over the rest of its matching preferences including $c_2$. Thus, $a_1$ would not break up with $c_1$ to match with $c_2$. 
As the proposed algorithm terminates either when all countermeasures are
matched to attacks or every unmatched countermeasure has been rejected
by every acceptable attack. Therefore, the algorithm terminates
after a finite number of steps.
A similar example can be provided for ASM algorithm.
As ASM and CSM algorithms respect Definition \ref{Blocked_Matching}, the proposed algorithms result in a stable matching for full attack coverage. Thus, the proof is completed.
\end{proof}

The partial coverage case may lead to an unstable matching but with higher coverage percentage, this will be significantly reduced. However, this problem can still be addressed, which is discussed in Section \ref{Pareto_Sec}. 

\begin{prop}
    \label{Multi_SM}
    There can be multiple potential matching solutions when $\sum_a N_a>\Bar{M}$
\end{prop}

\begin{proof}
\label{R3}
    Due to the nature of partial attack coverage of $\mathbf{C3.2}$ (i.e., the case $\sum_a N_a>\Bar{M}$) and the fact that the SM algorithms can be executed in different orders (i.e., starting the game from a different attack or countermeasure), the outcome can form a Pareto set of solutions.

    Here we provide an example to prove this. Let us assume $A=10$, $C=3$ and $\Bar{M}=80\%$. Let us for simplicity consider there are 10 nodes in the network that have each of these attacks. Fig.~\ref{SMExample1} shows the preference lists of the attacks and countermeasures. Solving the problem with different starting points (starting from $c_1$, $c_2$ and $c_3$) using Alg.~\ref{Count_Alg} we obtain the matching results as shown in Fig.~\ref{SMExample2}. As seen there can be three possible matchings of $M_1$, $M_2$, and $M_3$. It should be noted that Alg.~\ref{Count_Alg} stops when the number of covered attacks across the nodes reaches a minimum of 80\%. As seen, starting the CSM algorithm from different starting points results in different matching solutions. The same applies to the ASM algorithm.   
\end{proof}


\begin{figure}
\centering
\includegraphics[width=0.9\columnwidth]{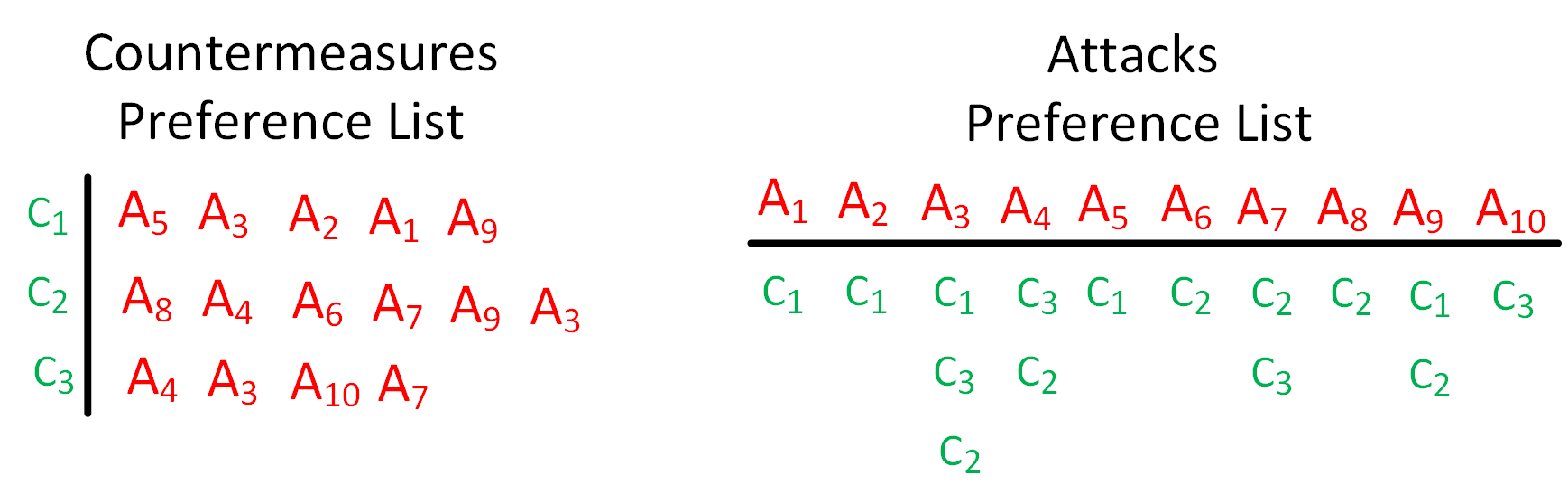}
\caption{Preference Lists of Countermeasures and attacks}
\label{SMExample1}
\end{figure}

\begin{figure}
\centering
\includegraphics[width=0.9\columnwidth]{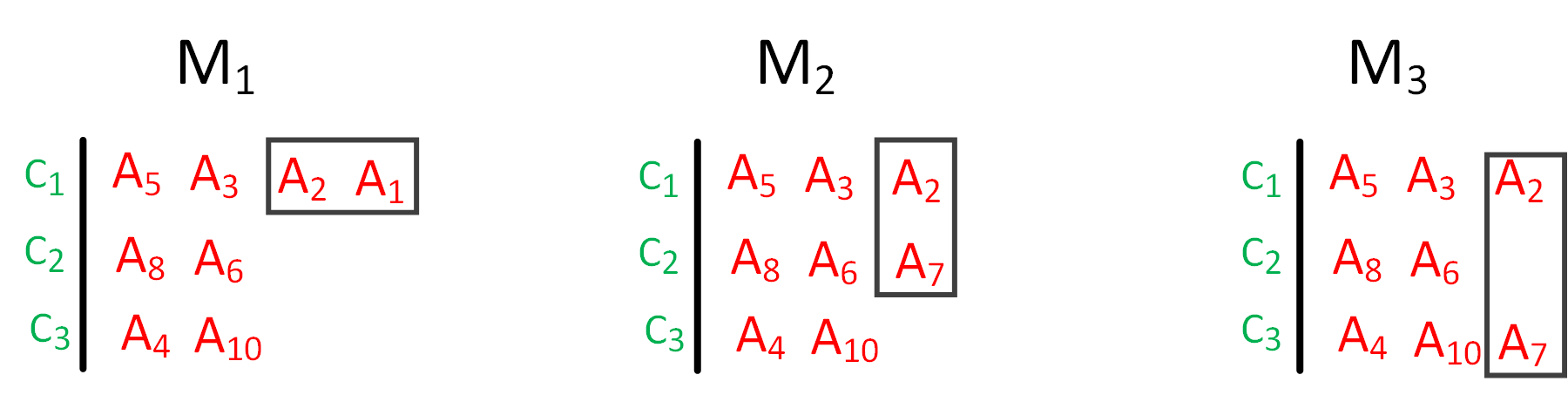}
\caption{SM Solutions}
\label{SMExample2}
\end{figure}

\section{simulation Results}
\label{sec:numericalresult}
We evaluate the performance of the proposed game theoretical-based methods by simulations performed in \textit{MATLAB}. Table~\ref{tab1} summarizes the simulation parameters. 
By analyzing several attack types targeting different networks (IP network, IoT, mobile networks, among others), their impact in terms of security, and the potential countermeasures to mitigate their impact, we noticed the presence of variability in the data. Therefore, in order to reflect this variability of inputs, we consider randomly generated values as given in Table~\ref{tab1} to avoid unrealistic {scenarios}\footnote{Realistic datasets for 5G core-related attacks are being collected in the context of the H2020 SANCUS project and will be shared with the community in the future. } for the evaluation of the proposed solutions. In the following, the performance of the two algorithms is evaluated by the impact of different parameters on the joint utility function, security utility, or QoS cost values. Please note that in the following sub-sections, we differentiate the terms \textit{attack} which is on a specific node, and the \textit{attack types} (e.g., DoS).


\begin{table}[tbp]
\centering
\caption{Simulation Setting}
\label{tab1}
\begin{tabular}{ | m{0.53\columnwidth} |  m{0.31\columnwidth}|} 
  \hline
  \textbf{Parameter} & \textbf{Value}\\
  \hline 
    \# of devices ($N$) &  100  \\
	\hline
  \# of attack types ($A$) &  [20  25  30  35  40]  \\
	\hline
  \# of countermeasure types ($C$) & [4  6  8  10  12]   \\
	\hline
	  Time, Energy, Monetary cost, Security & [0  1]  \\
	\hline
 	  \% of covered attacks ($\Bar{M}$) & [50  60  70  80  90  100]\%  \\
	\hline
 	  Monetary budget ($\xi$) & [4-12]  \\
	\hline
\end{tabular}
\end{table}

\subsection{Impact of $\beta$ on the QoS costs parameters}

Fig.~\ref{Beta} shows the impact of QoS coefficients on the cost of each of the QoS parameters (see Remark~\ref{R2}). This figure is the average result of 1000 simulation runs, where we relax $A=10$, $C=4$, $\bar{M}=90\%$, and $\xi=6$, respectively. As seen, in both Figs. \ref{VulBeta} and \ref{CountBeta} when $\beta_1$, $\beta_2$ and $\beta_3$ are set the highest value, i.e., 0.9, the lowest value of time, energy, and monetary cost, respectively, can be obtained. This is due to the impact of the coefficient in the matching result. For instance, when $\beta_1=0.9$, each attack type prioritizes \textit{time} more than the other QoS parameters for their preference list formation, which results in a more \textit{time-aware} matching solution. The same energy and cost minimization can be observed by setting $\beta_2$ and $\beta_3$ to the highest weight. However, when the coefficients are equal (the last set of bars), the obtained SM solutions have the same time, energy, and monetary cost values too. Finally, the SM solutions obtained from the ASM algorithm have slightly lower (i.e., better) QoS costs than the CSM algorithm. This is because the CSM algorithm prioritizes security more than QoS.

\begin{figure}[!tb]
\centering
\subfloat[ASM Alg.]{\includegraphics[width=\columnwidth]{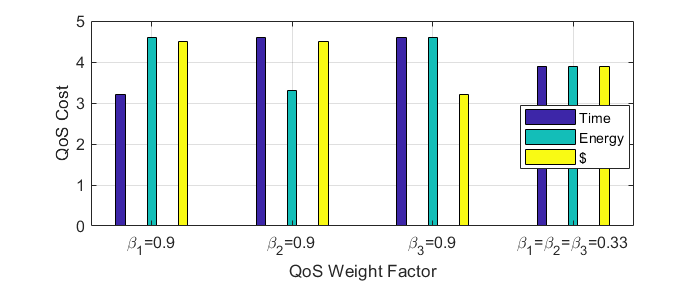}\label{VulBeta}}\\
\subfloat[CSM Alg.]{\includegraphics[width=\columnwidth]{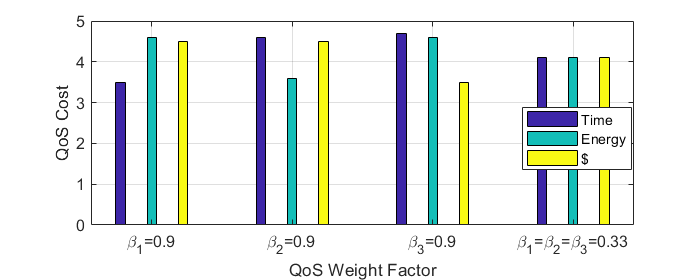}\label{CountBeta}}
\caption{Impact of $\beta$ on the QoS cost of the two algorithms}
\label{Beta}
\end{figure}

\subsection{Impact of monetary budget on the QoS/Monetary Costs and Security Utility}

\begin{figure}
\centering
\subfloat[QoS/Monetary Cost]{\includegraphics[width=0.5\columnwidth]{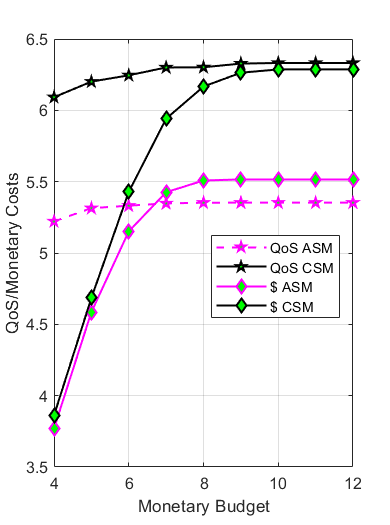}\label{MoneyQoS}}
\subfloat[Security Utility]{\includegraphics[width=0.5\columnwidth]{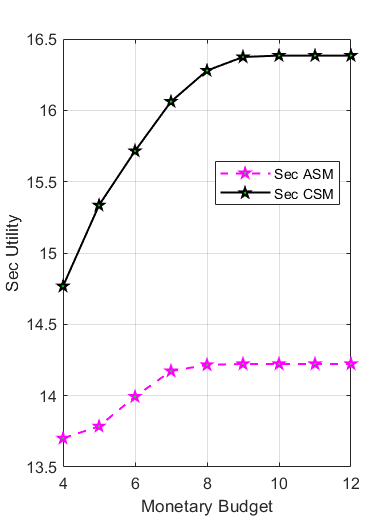}\label{MoneySec}}
\caption{Impact of monetary budget on the joint objective function, QoS, and security of the two algorithms}
\label{Budget}
\end{figure}

Figs.~\ref{MoneyQoS} and \ref{MoneySec} depict the impact of the monetary budget (see Constraint (\ref{c3.1})) on the QoS/Monetary costs and security utility, respectively, of the two ASM and CSM algorithms. This figure is the average result of 200 simulation runs, where we relax $A=20$, $C=10$, $\beta_i=0.33$, and $\bar{M}=100\%$.

The first observation is that the ASM algorithm has lower QoS costs than the CSM algorithm as it prioritizes the QoS for the matching (Figs.~\ref{MoneyQoS}). On the other hand, the CSM Algorithm outperforms the ASM Algorithm in terms of security utility as it prioritizes security for the matching Figs.~\ref{MoneySec}. The impact of monetary costs is directly reflected in both QoS costs and security utility as it restricts the optimization in finding a solution that optimizes the joint objective function in $\mathbf{P3}$. As seen in Fig.~\ref{MoneySec} as the monetary budget increases it brings higher options (larger solution pool) for the $\mathbf{P3}$; thus, higher security utility can also be obtained, however, it also increases the QoS costs, as seen in Fig.~\ref{MoneyQoS}. Therefore, this is a trade-off to be considered. Please note that the QoS costs, as defined in Eq.\eqref{U_V}, include time, energy, and monetary costs where each parameter is
multiplied by a $\beta_*=0.33$. The other observation is that even though the monetary costs are increasing up to 12, the matching algorithms do not select solutions with a cost higher than 8/10 for ASM/CSM algorithm as those solutions do not optimize $\mathbf{P3}$ (due to the higher QoS costs). These figures indicate that in order to cover $100\%$ of the attacks across the nodes while there are 20 different attack types and 10 countermeasures in hand, a monetary cost in the [4  10] range is needed, a {higher monetary cost}\footnote{A lower monetary cost does not allow for finding a feasible solution most of the times; thus, not suitable to consider} is not necessary. This implies if the monetary budget is restricted to 4, the maximum security that can be obtained from the solution is 13.7 and 14.7 for ASM and CSM algorithms.



\subsection{Impact of number of attack types and countermeasure types on the utility and cost values}

Figs. \ref{Ave3D_Vul} and \ref{Ave3D_Count} depict the average joint utility per attack (across the nodes) when impacted by various numbers of attack and countermeasure types. Each of the points in these figures represents the average of 200 simulation experiments where in each experiment random security and QoS values are generated for fairness. In order to focus on the impact of the number of attack and countermeasure types, we relax the $\beta_i=0.33$, $\bar{M}=90\%$, and $\xi=15$. This experiment answers the question "Do we receive a higher joint utility for each attack (across nodes) if there are more countermeasure and attack types in the network?". As seen, increasing the number of countermeasure types  increases the average attack (across the nodes) utility and by the increase in the number of attack types, the average attack (across nodes) utility remains quite stable. In order to better understand the reason we have plotted Fig.~\ref{Vul_Count}.

Fig~\ref{Vul_Count} depicts the impact of the number of attack and countermeasure types on the security and QoS of the solutions obtained by the two algorithms, where the results show the average of 200 simulation experiments. By taking a closer observation on fig~\ref{AveVC_2D_Vul} and \ref{AveVC_2D_Count} we can see that as the number of countermeasure types increases, there will be lower QoS costs per attack. This is because each attack type has wider options to choose from (or be chosen for the CSM algorithm); hence, a higher chance to match to a countermeasure type with lower QoS cost. Increasing the number of countermeasure types also increases the security per attack for the same reason. 


On the other hand, as the number of attack types in the system increases, the average security utility per attack in the system slightly decreases and QoS remains quite stable. This is because there are a fixed number of countermeasure types to address more attack types. However, as seen when the number of countermeasure types is 12, the QoS cost is the least, and security utility is the most in both Figs. \ref{AveVC_2D_Vul} and \ref{AveVC_2D_Count}. Finally, it can be observed that Fig.\ref{AveVC_2D_Count} depicts a higher security value as it is obtained by a CSM algorithm and Fig.\ref{AveVC_2D_Vul} has lower QoS costs as it is obtained by the ASM algorithm.

\begin{figure}[!tb]
\centering
\subfloat[ASM Alg.]{\includegraphics[width=0.85\columnwidth]{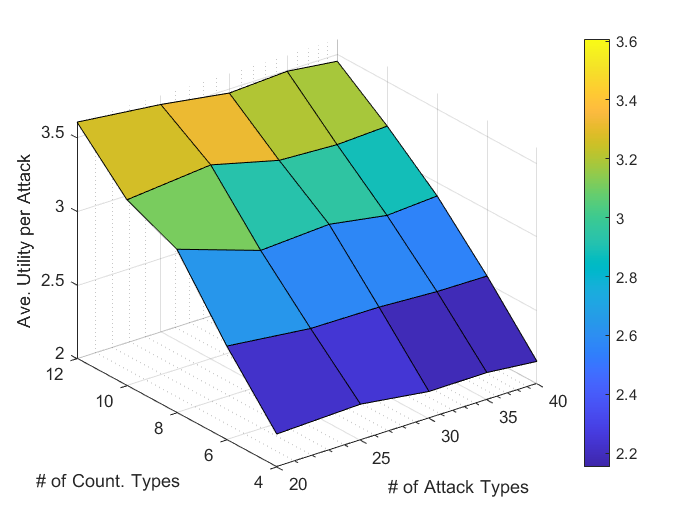}\label{Ave3D_Vul}}\\
\subfloat[CSM Alg.]{\includegraphics[width=0.85\columnwidth]{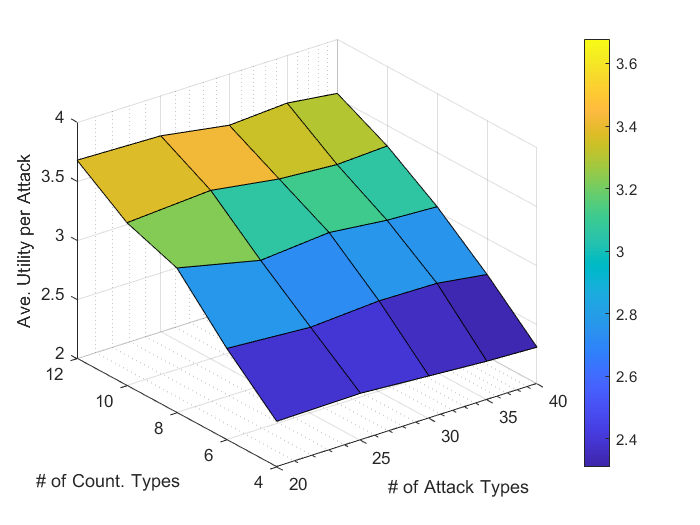}\label{Ave3D_Count}}
\caption{Impact of number of attack and countermeasure types on the utility (Per attack) of the two algorithms}
\label{3D}
\end{figure}

\begin{figure}[!tb]
\centering
\subfloat[ASM Alg.]{\includegraphics[width=0.95\columnwidth]{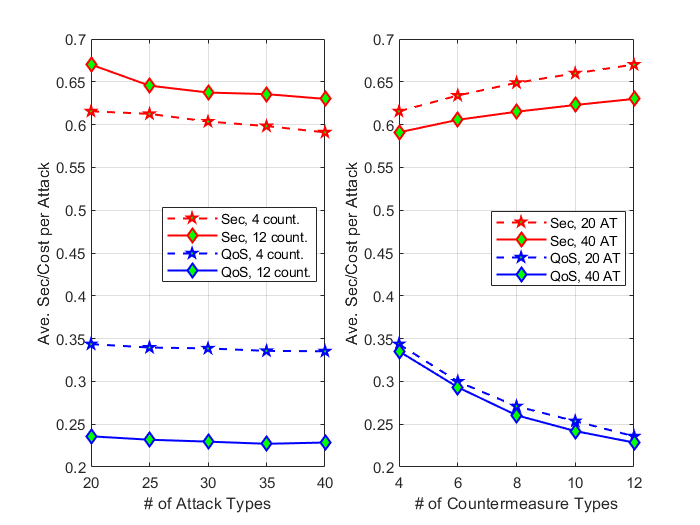}\label{AveVC_2D_Vul}}\\
\subfloat[CSM Alg.]{\includegraphics[width=0.95\columnwidth]{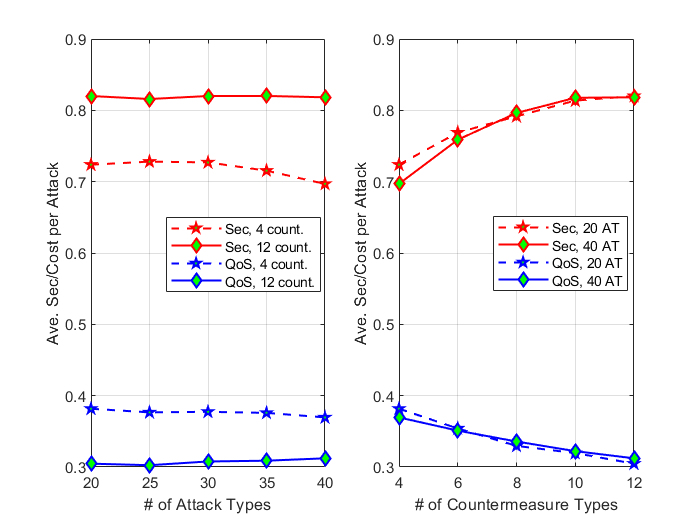}\label{AveVC_2D_Count}}
\caption{Impact of number of attack and countermeasure types on per attack security and QoS of the two algorithms}
\label{Vul_Count}
\end{figure}

\subsection{Impact of the percentage of covered attacks on objective function}

Fig.~\ref{Coverage} depicts the impact of Constraint (\ref{c3.2}) on the joint objective, QoS, and security of the two ASM and CSM algorithms. This figure is the average result of 500 simulation experiments, where we relax $A=10$, $C=4$, $\xi=6$, and $\beta_i=0.33$.

There are mainly two points that can be observed from the figures. First, as the percentage of covered attacks (across the nodes) increases, there is higher QoS costs, higher security utility, and higher joint utility value, which is expected. Second, the ASM algorithm has lower QoS costs due to the priority given to the QoS when performing the matching, and the CSM algorithm has higher security utility due to the priority given to the security objective when performing the matching. In the joint objective figure, however, the CSM algorithm performs better when it considers the ratio of security utility and QoS costs. Remarkably, the trade-off can be clearly observed as the graph in Fig.\ref{CoverageCost} (QoS costs) complements the graph in Fig.\ref{CoverageSec} (security Utility), resulting in a perfect matching. Finally, an SM is guaranteed for the case of $\Bar{M}=100\%$.

\begin{figure}[!tb]
\centering
\subfloat[Joint Objective Utility]{\includegraphics[width=0.95\columnwidth]{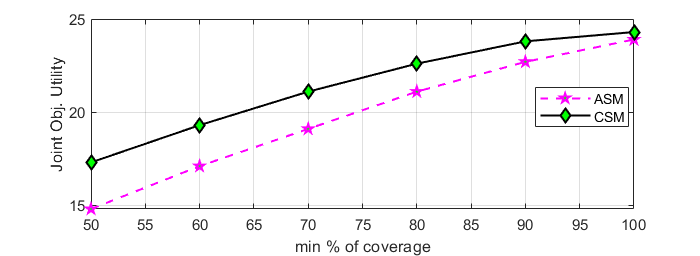}\label{CoverageJoint}}\\
\subfloat[QoS Cost]{\includegraphics[width=0.95\columnwidth]{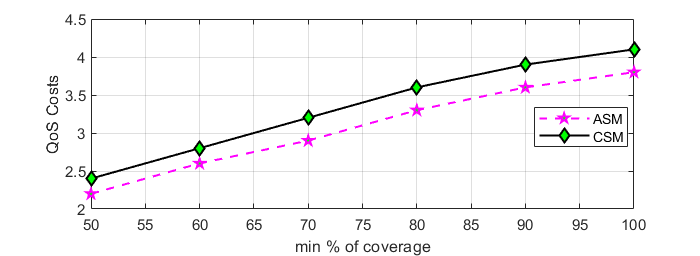}\label{CoverageCost}}\\
\subfloat[Security Utility]{\includegraphics[width=0.95\columnwidth]{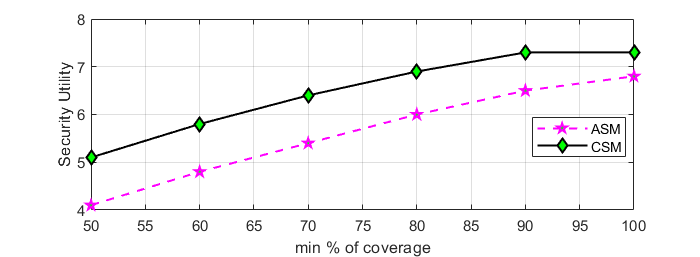}\label{CoverageSec}}
\caption{Impact of \% of covered attacks on the joint objective function, QoS costs, and security utility of the two algorithms}
\label{Coverage}
\end{figure}






\subsection{Pareto Front Solutions}
\label{Pareto_Sec}


\begin{figure*}
\centering
\subfloat[Feasible Solution 1]{\includegraphics[width=0.7\columnwidth,height=6cm]{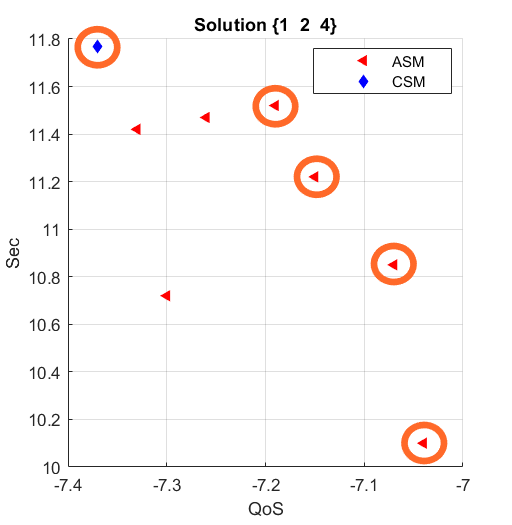}\label{Pareto1}}
\subfloat[Feasible Solution 2]{\includegraphics[width=0.7\columnwidth,height=6cm]{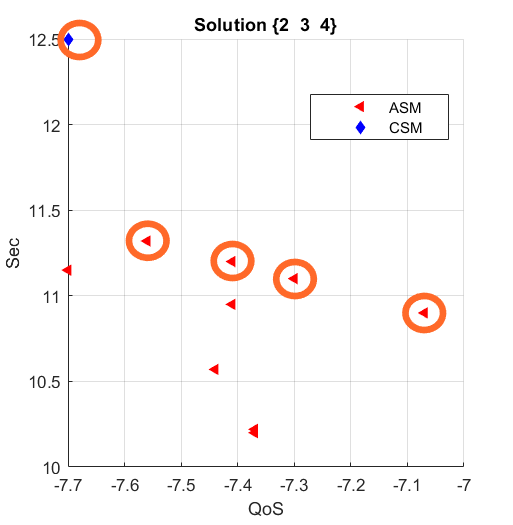}\label{Pareto2}}
\subfloat[Feasible Solution 3]{\includegraphics[width=0.7\columnwidth,height=6cm]{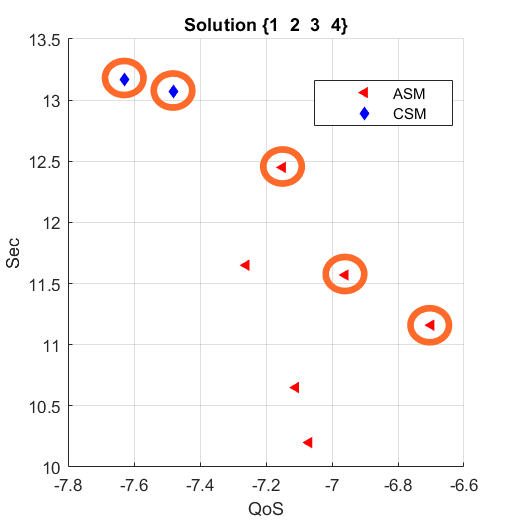}\label{Pareto3}}
\caption{Pareto set of Fronts of the two algorithms for each feasible solution}
\label{Pareto_FeasibleSet}
\end{figure*}


As discussed in  Proposition~\ref{Multi_SM}, for the case of partial attack coverage (across the nodes), the obtained solution might not be stable; hence, the algorithms might produce several solutions each by executing $\mathbf{P3}$ from different starting points. Each of these solutions might result in a better objective than the other (i.e., security or QoS). 
In this sub-section, we perform simulation results for this special case. In such a case, the quality of a solution can be determined by its Pareto-dominance with respect to other solutions \cite{Bozorgchenani_TMC21}. In particular, let $\boldsymbol{\Phi}=\{\Phi_1, \Phi_2, ..., \Phi_Z\}$ be the set of solutions (from the $z$th starting point), where $\Phi_z$ is the $z$th solution, and $Z$ is the total number of generated solutions. Considering two solutions, say $\Phi_1$ and $\Phi_2$, for a given problem with $S$ conflicting objectives, say $O_s$ (for all $s\in[1,S]$), we define Pareto-dominance as follows:

\begin{definition}
\label{Pareto_Def}
\emph{Let $O_s(\Phi)$ be the value of the objective function for the $s$-th objective evaluated at some solution $\Phi$. Then $\Phi_1$ is said to Pareto-dominate $\Phi_2$ (i.e., $\Phi_1\succ\Phi_2$) if $O_s(\Phi_1)\leq O_s(\Phi_2)$ for all $s\in [1,S]$, and there exists some
$p\in[1,S]$ such that $O_p(\Phi_1)<O_p(\Phi_2)$.}
\end{definition}
In our problem, $S=2$ represents the two security and QoS objectives, and $O_s$ represents the value of the objectives as defined in \eqref{U_V} and \eqref{U_C}.
We have set $A=20$, $C=4$, $\xi=7$, and $\Bar{M}=80\%$. Fig.\ref{Pareto_FeasibleSet} depicts three solution sets (for $\mathbf{P3}$) composed of different atomic countermeasures, i.e, \{1,2,4\}, \{2,3,4\} and \{1,2,3,4\} to be taken to address the attacks. These are the only feasible solution sets to cover a minimum of $80\%$ of the attacks (across the nodes) with different objective values, where each of these feasible solution sets can result in different solutions when executing $\mathbf{P3}$ from different starting points.

The red and blue points represent these different solutions (i.e., $\Phi_z$ as defined in Def.\ref{Pareto_Def}) for ASM and CSM algorithms. 
Fig.\ref{Pareto_FeasibleSet} represents only those solutions whose monetary cost does not violate the monetary budget in $\mathbf{C3.1}$. As seen the countermeasure-oriented solutions have the highest security utility and the attack-oriented solutions have the best QoS. Please note that QoS costs are negated in this figure for a better representation of the goodness of the Pareto Fronts. The solutions shown in the orange circle are the (strong) Pareto optimal solutions (any change makes at least one objective worse off), where they offer either lower QoS cost or higher security utility.


\section{Conclusion}
\label{sec:conclusion}
This work studies the countermeasure selection problem as part of an Intrusion Response System (IRS) by considering a trade-off between Security and QoS. The joint problem is formulated considering the constraints on monetary costs and the percentage of covered detected attacks. The problem is transformed into a game-theoretical model and addressed with a Stable Matching solution that considers the utility of two sides of the game, which are the attack and countermeasure types. We first derived the upper bound for the problem and later proposed algorithms to solve the game. Extensive simulation results are carried out to validate the performance of the game-theoretical solutions to see the impact of monetary costs, percentage of covered attacks, number of attack and countermeasure types on the joint utility function. Moreover, the Pareto front solutions are plotted to show the diverse feasible solutions with respect to security and QoS objectives for the special case of non-stable solutions.

In the future, we wish to extend this work by investigating the deployment order of countermeasures. This not only impacts the response effectiveness in terms of risk reduction but also impacts the time model. We also plan to study the execution duration of the selected countermeasures and the network area they are applied as they will impact the system costs/utility too. 



%



\ifCLASSOPTIONcaptionsoff
  \newpage
\fi

\bibliographystyle{IEEEtran}
\bibliography{IEEEabrv,Biblio}
\end{document}